\RequirePackage{fix-cm}
\documentclass[smallextended]{svjour3}       
\smartqed  
\usepackage{graphicx}
\usepackage{amssymb}
\usepackage{amsmath}
\begin{document}

\title{A new class of  hyper-bent
Boolean functions in binomial forms
\thanks{Chungming Tang, Yanfeng Qi and Maozhi Xu
acknowledge support from
the Natural Science Foundation of China
(Grant No.10990011 \& No.60763009). Baocheng Wang and Yixian
Yang acknowledge support from   National Science Foundation 
of China Innovative Grant (70921061), the CAS/SAFEA 
International Partnership Program for Creative Research Teams.
}
}


\author{Chunming Tang \and
 Yanfeng Qi \and Maozhi Xu \and  Baocheng Wang
\and Yixian Yang         
}


\institute{
 Chunming Tang,
 Yanfeng Qi and Maozhi Xu \at
              Laboratory of Mathematics and Applied Mathematics, School of Mathematical Sciences, Peking University, Beijing, 100871, China\\
              Chunming Tang's e-mail: tangchunmingmath@163.com\\
              Baocheng Wang  and  Yixian Yang\at
             Information Security Center, Beijing University of Posts and Telecommunications
 and Research Center on fictitious Economy and Data Science,
 Chinese Academy of Sciences,  Beijing, 100871,
 China 
}

\date{Received: date / Accepted: date}

\maketitle

\begin{abstract}
Bent functions, which are maximally nonlinear Boolean functions with even numbers of variables
and whose Hamming distance to the set of all affine functions equals $2^{n-1}\pm 2^{\frac{n}{2}-1}$,
were introduced by Rothaus in 1976 when he
considered problems in combinatorics. Bent functions have been extensively studied due to their applications in cryptography, such as S-box, block cipher and stream cipher. Further, they have been applied to coding theory, spread spectrum and combinatorial design. Hyper-bent functions, as a special class of bent functions, were introduced by Youssef and Gong in 2001, which have stronger
properties and rarer elements. Many research focus on the construction of bent and hyper-bent functions. In this paper, we consider functions defined over $\mathbb{F}_{2^n}$ by
$f_{a,b}:=\mathrm{Tr}_{1}^{n}(ax^{(2^m-1)})
+\mathrm{Tr}_{1}^{4}(bx^{\frac{2^n-1}{5}})$, where $n=2m$, $m\equiv 2\pmod 4$, $a\in \mathbb{F}_{2^m}$ and $b\in\mathbb{F}_{16}$. When  $a\in \mathbb{F}_{2^m}$ and $(b+1)(b^4+b+1)=0$,
with the help of Kloosterman sums and the factorization of  $x^5+x+a^{-1}$, we present a
characterization of hyper-bentness of $f_{a,b}$.
Further, we use generalized Ramanujan-Nagell equations to characterize hyper-bent functions of $f_{a,b}$ in the case  $a\in\mathbb{F}_{2^{\frac{m}{2}}}$.

\end{abstract}

\section{Introduction}
\label{intro}
Bent functions are maximally nonlinear
Boolean functions with even numbers of  variables
whose Hamming distance to the set of all
affine functions equals $2^{n-1}\pm
2^{\frac{n}{2}-1}$.
These functions introduced by
Rothaus \cite{OSR} as interesting combinatorial objects
have been extensively studied for their applications not only in cryptography, but also
in coding theory  \cite{CP,SM} and combinatorial design. Some basic knowledge and recent
results on bent functions can be found
in \cite{CC,HGT,SM}. A bent function can be
considered as a Boolean function defined
over $\mathbb{F}_{2}^{n}$,
$\mathbb{F}_{2^m}\times \mathbb{F}_{2^m}~(n=2m)$ or $\mathbb{F}_{2^n}$. Thanks to the different structures of the vector space
$\mathbb{F}_{2}^n$ and the Galois field
$\mathbb{F}_{2^n}$, bent functions can be
well studied. Although some algebraic
properties of bent functions are well known,
the general structure of bent functions
on $\mathbb{F}_{2^n}$ is not clear yet. As a result, much research on bent functions on
$\mathbb{F}_{2^n}$ can be found in
\cite{APG,PG,PGK,JD,JH,HGA,RG,GL,GA,SM3,SM1,SM,SMT,SMG,NG}.
Youssef and Gong \cite{AG} introduced
a class of bent functions called hyper-bent
functions, which achieve the maximal minimum
distance to all the coordinate functions of all
bijective monomials (i.e., functions of the form
$\mathrm{Tr}_{1}^{n}(ax^i)+
\epsilon $, $\mathrm{gcd}
(i, 2^n-1)=1$). However, the definition
of hyper-bent functions was given
by Gong and Golomb \cite{GS} by
a property of the extend Hadamard transform
of Boolean functions. Hyper-bent functions
as special bent functions with strong properties
are hard to characterize and many related problems  are open. Much research give
the precise characterization of hyper-bent
functions in certain forms.

The complete classification of bent and hyper-bent functions is not yet achieved. The monomial bent functions in the
form $\mathrm{Tr}_{1}^{n}(ax^s)$ are
considered in \cite{APG,GL}. Leander
\cite{GL} described the necessary conditions for
$s$ such that $\mathrm{Tr}_{1}^{n}(ax^s)$ is a bent function. In particular, when
$s=r(2^m-1)$ and $(r,2^m+1)=1$, the monomial functions $\mathrm{Tr}_{1}^{n}(ax^s)$ (i.e.,
the Dillon functions) were extensively studied
in \cite{PG,JD,GL}. A class of quadratic functions over $\mathbb{F}_{2^n}$ in polynomial form $\sum\limits_{i=1}^{\frac{n}{2}
-1}a_{i}\mathrm{Tr}_{1}^{n}(x^{1+2^i})
+a_{\frac{n}{2}}\mathrm{Tr}_{1}^{\frac{n}{2}}
(x^{\frac{n}{2}+1})~(a_{i}\in \mathbb{F}_{2})$ was described and studied in \cite{PEC,HD,TK,SJ,WMF,NG}.
Dobbertin et al. \cite{HGA}
constructed a class of binomial bent functions
of the form $\mathrm{Tr}_{1}^{n}(a_{1}x^{s_{1}}
+a_{2}x^{s_{2}}), (a_{1},a_{2})
\in (\mathbb{F}_{2^n}^{*})^2$
with  Niho power functions. Garlet
and Mesanager \cite{CS} studied the duals of
the Niho bent functions in \cite{HGA}.
In \cite{SM3,SM1,SMG}, Mesnager considered
the binomial functions of the form
$\mathrm{Tr}_{1}^{n}(ax^{r(2^{m}-1)})
+\mathrm{Tr}_{1}^{2}(bx^{\frac{2^n-1}{3}})$,
where $a\in \mathbb{F}_{2^n}^*$ and $b\in
\mathbb{F}_{4}^{*}$. Then he gave the link between the bentness property of such functions and Kloosterman sums.
Leander and Kholosha \cite{GA} generalized
one of the constructions proven by Dobbertin et al. \cite{HGA} and presented a new primary construction of bent functions consisting of a linear combination of $2^r$ Niho exponents.
Carlet et al. \cite{CTA} computed the dual of
the Niho bent function with $2^r$ exponents found
by Leander and Kholosha \cite{GA} and
showed that this new bent function is not
of the Niho type.  Charpin and Gong \cite{PG}
presented a characterization of bentness of Boolean functions over $\mathbb{F}_{2^n}$ of the form
$\sum\limits_{r\in R}\mathrm{Tr}_{1}^{n}
(a_{r}x^{r(2^m-1)})$, where $R$ is a subset of the set of representatives of the cyclotomic cosets modulo $2^m+1$ of maximal size $n$. These functions include the well-known monomial functions with the Dillon exponent as a special case. Then they
described the bentness of these functions with the Dickson polynomials.
Mesnager et al. \cite{SM,SMT} generalized the results of Charpin and Gong \cite{PG} and considered the bentness of Boolean functions over $\mathbb{F}_{2^n}$ of the form $\sum\limits_{r\in R}\mathrm{Tr}_{1}^{n}(a_{r}x^{r(2^m-1)})
+\mathrm{Tr}_{1}^{2}(bx^{\frac{2^n-1}{3}})$, where
$n=2m$, $a_{r}\in \mathbb{F}_{2^m}$ and  $b\in\mathbb{F}_{4}$. Further, they presented the link between the bentness of such functions and some exponential sums (involving Dickson polynomials).

In this paper, we consider a class of
Boolean functions $\mathcal{H}_{n}$.
A Boolean function $f_{a,b}$ in $\mathcal{H}_{n}$ is defined over $\mathbb{F}_{2^n}$ by the form:
$f_{a,b}:=\mathrm{Tr}_{1}^{n}(ax^{(2^m-1)})
+\mathrm{Tr}_{1}^{4}(bx^{\frac{2^n-1}{5}})$, where $n=2m$, $m\equiv 2\pmod 4$, $a
\in \mathbb{F}_{2^m}$ and $b\in\mathbb{F}_{16}$.
When $b=0$, Charpin and Gong \cite{PG} described the bentness and hyper-bentness of these functions with some character sums involving Dickson polynomial. Generally, it is elusive to give a characterization of bentness and hyper-bentness of Boolean functions in $\mathcal{H}_{n}$.
This paper discusses the hyper-bentness of
Boolean function in $\mathcal{H}_{n}$ for two cases. In the  first case $b=1$ or $b^4+b+1=0$,
we present the hyper-bentness of $f_{a,b}$ by
the factorization of $x^5+x+a^{-1}$ and Kloosterman sums. For the second case $a\in\mathbb{F}_{2^{\frac{m}{2}}}$, we give all the hyper-bent functions with the help of generalized Ramanujan-Nagell equations.

The rest of paper is organized as follows.
In Section 2, we give some notations and recall
some basic knowledge for the paper. In Section 3, we study the hyper-bentness of Boolean functions in $\mathcal{H}_{n}$ for two cases (1) $b=1$ or $b^4+b+1=0$;
(2) $a\in\mathbb{F}_{\frac{m}{2}}$. Finally, Section 4 makes a conclusion.

\section{Preliminaries}
\label{sec:2}
\subsection{Boolean functions}
\label{subsec:2.1}
Let $n$ be a positive integer. $\mathbb{F}_{2}^n$ is a n-dimensional vector space defined over finite field $\mathbb{F}_{2}$.
Take two vectors
$x=(x_{1},\cdots,x_{n})$ and $y=(y_{1},\cdots, x_{n})$ in $\mathbb{F}_{2}^n$. Their dot product is defined by
$$
\langle x,y\rangle:=\sum_{i=1}^{n}x_{i}y_{i}.
$$
$\mathbb{F}_{2^n}$ is a finite field
with $2^n$ elements and $\mathbb{F}_{2^n}^*$ is the multiplicative group of $\mathbb{F}_{2^n}$.
Let $\mathbb{F}_{2^k}$ be a subfield of
$\mathbb{F}_{2^n}$. The trace function from
$\mathbb{F}_{2^n}$ to $\mathbb{F}_{2^k}$,
denoted by $\mathrm{Tr}_{k}^{n}$, is a map
defined as
$$
\mathrm{Tr}_{k}^{n}(x):
= x+ x^{2^k}+
x^{2^{2k}}+\cdots+ x^{2^{n-k}}.
$$
When $k=1$, $\mathrm{Tr}_{1}^{n}$ is called the absolute trace. The trace function $\mathrm{Tr}_{k}^{n}$ satisfies the following
properties.
\begin{align*}
&\mathrm{Tr}_{k}^{n}(ax+by)=
a\mathrm{Tr}_{k}^{n}(x)+b\mathrm{Tr}_{k}^{n}(y),
~~a,b\in \mathbb{F}_{2^k}, x,y\in \mathbb{F}_{2^n}. \\
& \mathrm{Tr}_{k}^{n}(x^{2^k})
=\mathrm{Tr}_{k}^{n}(x), ~~x\in \mathbb{F}_{2^n}.
\end{align*}
When $\mathbb{F}_{2^k}\subseteq \mathbb{F}_{2^r}
\subseteq \mathbb{F}_{2^n}$, the trace function
$\mathrm{Tr}_{k}^{n}$ satisfies the following  transitivity property.
$$
\mathrm{Tr}_{k}^{n}(x)=\mathrm{Tr}_{k}^{r}(
\mathrm{Tr}_{r}^{n}(x)), ~~~x\in \mathbb{F}_{2^n}.
$$
A Boolean function over $\mathbb{F}_{2}^n$ or $\mathbb{F}_{2^n}$ is an $\mathbb{F}_{2}$-valued function. The absolute trace function is  a useful tool in constructing Boolean functions over
$\mathbb{F}_{2^n}$. From the absolute trace function, a dot product  over
$\mathbb{F}_{2^n}$ is defined by
$$
\langle x,y\rangle:= \mathrm{Tr}_{1}^{n}(xy),
~~~x,y\in \mathbb{F}_{2^n}.
$$
A Boolean function over $\mathbb{F}_{2^n}$ is often
represented by the algebraic normal form (ANF):
$$
f(x_{1}, \cdots, x_{n})=
\sum_{I\subseteq \{1,\cdots,n\}}
a_{I}(\prod_{i\in I}x_{i}),~~~ a_{I}\in \mathbb{F}_{2}.
$$
When $I=\emptyset$, let $\prod\limits_{i\in I}=1$. The terms $\prod\limits_{i\in I}x_{i}$ are called
monomials. The algebraic degree of
a Boolean function $f$ is the globe degree of
its ANF, that is, $\mathrm{deg}(f):=\mathrm{max}
\{\#(I)|a_{I}\neq 0\}$, where $\#(I)$ is the order of $I$ and $\#(\emptyset)=0$.

Another representation of a Boolean function is
of the form
$$
f(x)=\sum_{j=0}^{2^n-1}a_{j}x^j.
$$
In order to make $f$ a Boolean function, we should require $a_{0}, a_{2^n-1}\in\mathbb{F}_{2}$
 and $a_{2j}=a_{j}^2$, where
$2j$ is taken modulo $2^n-1$.
This makes that $f$ can be represented by
a trace expansion of the form
$$
f(x)=\sum_{j\in \Gamma_{n}}
\mathrm{Tr}_{1}^{o(j)}(a_{j}x^j)+
\epsilon(1+x^{2^{n}-1})
$$
called its polynomial form, where
\begin{itemize}
\item $\Gamma_n$ is the set of integers obtained by choosing one element in each cyclotomic class of 2 module
    $2^n-1$ ($j$ is often chosen as the smallest element in its cyclotomic class, called the coset leader of the class);
\item $o(j)$ is the size of the cyclotomic coset of 2 modulo $2^n-1$ containing $j$;
\item $a_{j}\in\mathbb{F}_{2^{o(j)}}$;
\item $\epsilon=wt(f)\pmod 2$, where
$\mathrm{wt}(f):=\#\{
x\in\mathbb{F}_{2^n}|f(x)=1
\}$.
\end{itemize}
Let $\mathrm{wt}_{2}(j)$ be the number
of 1's in its binary expansion. Then
$$
\mathrm{deg}(f)=
\left\{
  \begin{array}{l}
    n, ~~~~~~~~~~~~~~~~~~~~~~~~~\epsilon=1 \\
    \mathrm{max}\{
\mathrm{wt}_{2}(j)|a_{j}\neq 0\}, ~~\epsilon=0.
  \end{array}
\right.
$$
\subsection{Bent and hyper-bent functions}
\label{subsec:2.2}
The "sign" function of a Boolean function $f$ is defined by
$$
\chi(f):=(-1)^f.
$$
When $f$ is a Boolean function over
 $\mathbb{F}_{2}^n$, the Walsh Hadamard transform of $f$ is the discrete Fourier transform
of $\chi(f)$, whose value at
$w\in \mathbb{F}_{2}^n$ is defined by
$$
\widehat{\chi}_{f}(w):=
\sum_{x\in\mathbb{F}_{2}^n}(-1)^{f(x)
+\langle w,x\rangle}.
$$
When $f$ is a Boolean function over  $\mathbb{F}_{2^n}$, the Walsh Hadamard
transform of $f$ is defined by
$$
\widehat{\chi}_{f}(w):=
\sum_{x\in\mathbb{F}_{2^n}}(-1)^{f(x)
+\mathrm{Tr}_{1}^{n}(wx)},
$$
where $w\in\mathbb{F}_{2^n}$.
Then we can define the bent functions.
\begin{definition}\label{defn2.1}
A Boolean function $f:\mathbb{F}_{2^n}\rightarrow
\mathbb{F}_{2}$ is called a bent function, if
$\widehat{\chi}_{f}(w)
=\pm2^{\frac{n}{2}}~ (\forall w\in \mathbb{F}_{2^n})$.
\end{definition}
If $f$ is a bent function, $n$ must be
even. Further, $\mathrm{deg}(f)\leq \frac{n}{2}$ \cite{CC}. Hyper-bent functions are an important  subclass of bent functions. The definition of hyper-bent functions is given below.
\begin{definition}\label{defn2.2}
A bent function $f:\mathbb{F}_{2^n}\rightarrow \mathbb{F}_{2}$ is called a hyper-bent function, if,  for any $i$ satisfying $(i,2^n-1)=1$, $f(x^i)$ is also a bent function.
\end{definition}
\cite{CP} and \cite{AG} proved that
if $f$ is a hyper-bent function, then
$\mathrm{deg}(f)
=\frac{n}{2}$. For a bent function
$f$, $\mathrm{wt}(f)$ is even.  Then
$\epsilon=0$, that is,
$$
f(x)=\sum_{j\in\Gamma_{n}}\mathrm{Tr}_{1}
^{o(j)}(a_{j}x^j).
$$
If  a Boolean function $f$ is defined on
$\mathbb{F}_{2^{\frac{n}{2}}}\times
\mathbb{F}_{2^{\frac{n}{2}}}$, then we have a class of bent functions\cite{JD,RLM}.
\begin{definition}\label{defn2.3}
The Maiorana-McFarland class $\mathcal{M}$
is the set of all the Boolean functions $f$
defined on $\mathbb{F}_{2^{\frac{n}{2}}}\times
\mathbb{F}_{2^{\frac{n}{2}}}$ of the form $
f(x,y)=\langle x,\pi(y)\rangle +g(y)$,
where $x,y\in
\mathbb{F}_{2^{
\frac{n}{2}}}$, $\pi$ is a permutation of $\mathbb{F}_{2^{\frac{n}{2}}}$ and
$g(x)$ is a Boolean function over  $\mathbb{F}_{2^{\frac{n}{2}}}$.
\end{definition}
For Boolean functions over
$\mathbb{F}_{2^{\frac{n}{2}}}\times
\mathbb{F}_{2^{\frac{n}{2}}}$, we have a class of hyper-bent functions $\mathcal{PS}_{ap}$ \cite{CP}.
\begin{definition}\label{defn2.4}
Let $n=2m$, the $\mathcal{PS}_{ap}$ class is the set of all the Boolean functions of the form
$f(x,y)=g(\frac{x}{y})$, where $x,y\in\mathbb{F}_{2^m}$ and
$g$ is a balanced Boolean functions (i.e., $\mathrm{wt}(f)=2^{m-1}$) and $g(0)=0$.
When $y=0$, let $\frac{x}{y}
=xy^{2^{n}-2}=0$.
\end{definition}
Each Boolean function $f$ in $\mathcal{PS}_{ap}$
satisfies $f(\beta z)=f(z)$ and
$f(0)=0$, where $\beta\in\mathbb{F}_{m}^{*}$ and
$z\in\mathbb{F}_{m}\times\mathbb{F}_{m}$.
Youssef and Gong \cite{AG} studied these
functions over $\mathbb{F}_{2^n}$ and gave the following property.
\begin{proposition}\label{prop2.1}
Let $n=2m$, $\alpha$ be a primitive element in $\mathbb{F}_{2^n}$ and  $f$ be a Boolean function over $\mathbb{F}_{2^n}$ such that $f(\alpha^{2^m+1} x)
=f(x) (\forall x\in\mathbb{F}_{2^n})$
and $f(0)=0$, then $f$ is a hyper-bent function if and only if the weight of $(f(1)$,$f(\alpha)$, $f(\alpha^2)$,$\cdots$,
$f(\alpha^{2^m}))$ is $2^{m-1}$.
\end{proposition}
Further, \cite{CP} proved the following result.
\begin{proposition}\label{prop2.2}
Let $f$ be a Boolean function defined in Proposition \ref{prop2.1}. If $f(1)=0$, then $f$ is in $\mathcal{PS}_{ap}$. If $f(1)=1$, then there exists a Boolean function $g$ in $\mathcal{PS}_{ap}$ and
$\delta\in \mathbb{F}_{2^n}^{*}$ satisfying
$f(x)=g(\delta x)$.
\end{proposition}
Let $\mathcal{PS}_{ap}^{\#}$ be the set of hyper-bent functions in the form of
$g(\delta x)$, where
$g(x)\in\mathcal{PS}_{ap}$,
$\delta\in\mathbb{F}_{2^n}^*$ and
$g(\delta)=1$.
Charpin and Gong expressed Proposition \ref{prop2.2} in a different version below.
\begin{proposition}\label{prop2.3}
Let $n=2m$, $\alpha$ be a primitive element of
$\mathbb{F}_{2^n}$ and $f$ be a Boolean function
over  $\mathbb{F}_{2^n}$ satisfying $f(\alpha^{2^{m+1}}x)
=f(x)~(\forall x\in\mathbb{F}_{2^n})$ and
$f(0)=0$. Let $\xi$ be a primitive
$2^m+1$-th root in $\mathbb{F}_{2^n}^*$.
Then $f$ is a hyper-bent function if and only if the cardinality of the set
~$\{i|f(\xi^i)=1,
0\leq i\leq 2^m\}$ is $2^{m-1}$.
\end{proposition}
In fact, Dillon \cite{JD} introduced a
bigger class of bent functions the
Partial Spreads class $\mathcal{PS}^{-}$ than
$\mathcal{PS}_{ap}$ and $\mathcal{PS}_{ap}^{\#}$.
\begin{theorem}\label{thm2.1}
Let $E_{i}(i=1,2,\cdots,N)$ be $N$ subspaces in  $\mathbb{F}_{2^n}$ of dimension $m$ such that  $E_{i}\cap E_{j}=\{0\}$ for all
$i,j\in\{1,\cdots,N\}$ with
$i\neq j$. Let $f$ be a Boolean function over  $\mathbb{F}_{2^n}$. If the support of $f$ is given by $supp(f)=\bigcup\limits_{i=1}^{N}E_{i}^{*}$,
where $E_{i}^{*}=E_{i}\backslash \{0\}$, then
$f$ is a bent function if and only if $N=2^{m-1}$.
\end{theorem}
The set of all the functions in Theorem \ref{thm2.1} is defined by $\mathcal{PS}^{-}$.

\subsection{Dickson polynomials}
Now we recall the knowledge of Dickson polynomials over  $\mathbb{F}_{2}$ \cite{GRG}. For $r>0$, Dickson
polynomials are given by
$$
D_{r}(x)=\sum_{i=0}^{\lfloor \frac{r}{2}\rfloor}
\frac{r}{r-i}{{r-i \choose i}}x^{r-2i}, r=2,3,\cdots .
$$
Further, Dickson polynomials can be also defined by the following recurrence relation.
$$
D_{i+2}(x)=xD_{i+1}+D_{i}(x)
$$
with initial values
$$
D_{0}(x)=0, D_{1}(x)=x.
$$
Some properties of Dickson polynomials are given below.
\begin{itemize}
\item $\mathrm{deg}(D_{r}(x))=r$.
\item $D_{rp}(x)=D_{r}(D_{p}(x))$.
\item $D_{r}(x+x^{-1})=x^{r}+x^{-r}$.
\end{itemize}
The first few Dickson polynomials with odd $r$ are
\begin{align*}
&D_{1}(x)=x,\\
&D_{3}(x)=x+x^3,\\
&D_{5}(x)=x+x^3+x^5,\\
&D_{7}(x)=x+x^5+x^7,\\
&D_{9}(x)=x+x^5+x^7+x^9,\\
&D_{11}(x)=x+x^3+x^5+x^9+x^{11}.
\end{align*}
The following proposition gives
some properties on Dickson polynomials
\cite{JH,GRG}.
\begin{proposition}\label{prop2.4}
Let $m$ and  $k$ be two positive integers.
Let $x_{0}, y_{0}\in\mathbb{F}_{2^m}$ and $y_{0}=D_{k}(x_{0})$, then
$$
\#\{x\in\mathbb{F}_{2^m}|D_{k}(x)=y_{0}\}=
\left\{
  \begin{array}{l}
    d_{1}\quad \text{If} ~ x^2+x_{0}x+1~\text{is irreducible over}~\mathbb{F}_{2^m}~\text{and}~
y_{0}\neq 0\\
    d_{2}\quad \text{If} ~ x^2+x_{0}x+1~\text{is reducible over}~\mathbb{F}_{2^m}~\text{and}~
y_{0}\neq 0 \\
\frac{d_{1}+d_{2}}{2}\quad \text{If}~y_{0}=0
  \end{array}
\right.
$$
where $d_{1}=(k,2^m-1)$ and $d_{2}=(k,2^m+1)$.
\end{proposition}
The reducibility of the polynomial $x^2+x_{0}x+1$ can be determined by the following proposition.
\begin{proposition}\label{prop2.5}
Let $m$ be a positive integer and  $x_{0}\in\mathbb{F}_{2^m}$. Then $x^2+x_{0}x+1$ is reducible over $\mathbb{F}_{2^m}$ if and only if $\mathrm{Tr}_{1}^{m}(\frac{1}{x_{0}})=0$.
\end{proposition}

\subsection{Kloosterman sums and Weil sums}
\label{subsec:2.4}
The Kloosterman sums on $\mathbb{F}_{2^n}$ are:
$$
K_m(a):=\sum_{x \in \mathbb{F}_{2^m}}
\chi(\mathrm{Tr}^m_1(ax+\frac{1}{x})), \quad a\in \mathbb{F}_{2^m}.
$$
Some properties of Kloosterman sums are given by
the following proposition\cite{Hel,LWo}.
\begin{proposition}\label{prop2.6}
Let $a\in \mathbb{F}_{2^m}$.  Then $K_m(a)\in [1-2^{(m+2)/2},1+2^{(m+2)/2}]$ and $4\mid K_m(a)$.
\end{proposition}
Weil sums of degree 5 on $\mathbb{F}_{2^m}$ are:
$$
Q_m(a):=\sum_{x \in \mathbb{F}_{2^m}}
\chi(\mathrm{Tr}^m_1(a(x^5+x^3+x))),
\quad a\in \mathbb{F}_{2^m}.
$$
To determine the value of $Q_{m}(a)$, we introduce an affine Artin-Schreier model of the form \cite{MNa,vv}:
$$
{C}: y^2+y=a(x^5+x^3+x),
a\in\mathbb{F}_{2^m}^{*}.
$$
This curve is a smooth supersingular curve
of genus 2 over $\mathbb{F}_{2^m}$ and
it has only one point at infinity.
Let $J(C)$ be the Jacobian of $C$ over
$\mathbb{F}_{2^m}$.
Let $f_{J(C)}(x)$ be the Weil polynomial for $J(C)$ of the form:
$$
f_{J(C)}(x)=x^4+rx^3+sx^2+2^mrx+2^n,
$$
where $(r,s)$ is determined by the
irreducible factors of the polynomial  $P(x)=x^5+x+a^{-1}$.
We write that $P(x)
=(n_{1})^{r_{1}}(n_{2})^{r_{2}}
\cdots(n_{t})^{r_{t}}$ to indicate that
$r_{i}$ of the irreducible factors of
$P(x)$ have degree $n_{i}$.
Some results on the Weil polynomial of $C$ are
given in the following proposition \cite{CNa}.
\begin{proposition}\label{prop2.7}
Let $P(x)$, $r$ and $s$ be defined above, then

{\rm (1)} If $m$ is even, $(r,s)$ is determined by Table $1$.
\begin{table}[htbp]
\caption{$P(x)$ and $(r,s)$ for even $m$}
\label{tab:1}       
\begin{tabular}{ll}
\hline\noalign{\smallskip}
$P(x)$ & $(r,s)$   \\
\noalign{\smallskip}\hline\noalign{\smallskip}
$(5)$ & $(\pm 2^{m/2}, 2^{m/2})$\\
$(1)^2(3)$ & $(\pm 2\cdot2^{m/2}, 3\cdot2^{m})$\\
$(1)(2)^2$ & $(0, 2\cdot2^{m})$\\
$(1)^5$ & $(\pm 4\cdot2^{m/2}, 6\cdot2^{m})$\\
\noalign{\smallskip}\hline
\end{tabular}
\end{table}

{\rm (2)} If $m$ is odd, $(r,s)$ is determined by Table $2$.
\begin{table}[htbp]
\caption{$P(x)$ and $(r,s)$ for odd $m$}
\label{tab:2}       
\begin{tabular}{ll}
\hline\noalign{\smallskip}
$P(x)$ & $(r,s)$ \\
\noalign{\smallskip}\hline\noalign{\smallskip}
$(1)(4)$ & $(0, 0)$\\
$(2)(3)$ & $(0, 2^{m})$\\
$(1)^3(2)$ & $(0, -2\cdot2^{m})$\\
\noalign{\smallskip}\hline
\end{tabular}
\end{table}
\end{proposition}
The number of rational points on $C$ (including the infinite point) can be determined by  the Weil polynomial $f_{J(C)}(x)$. We write
$N(\mathbb{F}_{2^m})=\#(C(\mathbb{F}_{2^m}))$.
Then \cite{MNa}
\begin{equation}\label{equation}
N(\mathbb{F}_{2^m})=2^m+1+r,
N(\mathbb{F}_{2^{2m}})=2^{2m}+1
+2s-r^2.
\end{equation}
Further, we have the following result.
\begin{proposition}\label{prop2.8}
Let $Q_m(a)$ and $r$ be defined above, then
$$
Q_m(a)=r.
$$
\end{proposition}
\begin{proof}
To prove $Q_m(a)=r$, we just prove that
$$Q_{m}(a)=N(F_{2^m})-(2^m+1)=\#(C(\mathbb{F}_{2^{m}}))
-(2^m+1).$$
$C(\mathbb{F}_{2^m})$ has only an infinite point.
A point $(x,y)$ is in $C(\mathbb{F}_{2^{m}})$ if and only if $x\in \mathbb{F}_{2^m}$ and $\mathrm{Tr}_{1}^{m}(a(x^5+x^3+x))=0$. Then
\begin{align*}
\#(C(\mathbb{F}_{2^{m}}))=&1+
2\cdot \#\{x\in \mathbb{F}_{2^m}|
\mathrm{Tr}_{1}^{m}(a(x^5+x^3+x))=0\}\\
=& 2\cdot \#\{x\in \mathbb{F}_{2^m}|
\mathrm{Tr}_{1}^{m}(a(x^5+x^3+x))=0\}
-2^{m}+(2^m+1)\\
=& \sum_{x\in\mathbb{F}_{2^m}}
\chi(\mathrm{Tr}_{1}^{m}(a(x^5+x^3+x)))
+(2^m+1).
\end{align*}
From the definition of $Q_{m}(a)$, $$\#(C(\mathbb{F}_{2^{m}}))
=2^m+1+Q_{m}(a).$$
Hence, $Q_{m}(a)=r$.
\end{proof}
From Proposition \ref{prop2.7} and Proposition \ref{prop2.8}, we can easily obtain two following
corollaries.
\begin{corollary}\label{cor2.1}
Let $m$ be even and  $a\in\mathbb{F}_{2^m}$, then $Q_m(a)\in\{0,\pm 2^{m/2},\pm 2\cdot2^{m/2},\pm 4\cdot2^{m/2}\}$.
\end{corollary}

\begin{corollary}\label{cor2.2}
Let $m$ be even and $a\in\mathbb{F}_{2^m}^{*}$,
then

{\rm (1)} $Q_{m}(a)=0$ if and only if
~$P(x)=x^5+x+a^{-1}=(1)(2)^2$.

{\rm (2)} $Q_{m}(a)=\pm 2^{m/2}$ if and only if $P(x)=x^5+x+a^{-1}$ is irreducible over $\mathbb{F}_{2^m}$.
\end{corollary}
\begin{remark}
When $P(x)=x^5+x+a^{-1}$ is irreducible over $\mathbb{F}_{2^m}$, the sign of $Q_{m}(a)$ is related to the parity of the quadratic form
$\mathfrak{q}(x)=\mathrm{Tr}_{1}^{m}
(x(ax^4+ax^2+a^2x))$.
$\mathfrak{q}(x)$ is the quadratic form associated to the simplectic form:
$$
<x,y>_{\mathfrak{q}}:=
\mathrm{Tr}_{1}^{m}(x(ay^4+ay^2+a^2y)
+y(ax^4+ax^2+a^2x)),
$$
which is non-degenerate. Then there exists
a normal simplectic basis $e_1$, $e_{m_1+1}$, $\cdots$, $e_{m_1}$, $e_{2m_1}$
$(2m_{1}=m)$.
If $i \not\equiv j (\mod m_1)$,
$<e_i,e_j>_{\mathfrak{q}}=0$. For any
$i~(1\leq i \leq m_1)$, $
<e_i,e_{m_1+i}>_{\mathfrak{q}}=1$.
If $\#\{i|\mathfrak{q}(e_i)=\mathfrak{q}(e_{m_1+i})=1, 1\leq i \leq m_1\}$ is even, then the quadratic form $\mathfrak{q}(x)$ is even and $Q_m(a)=2^{m_1}$. If $\#\{i|\mathfrak{q}(e_i)=
\mathfrak{q}(e_{m_1+i})=1,1\leq i \leq m_1\}$
is odd, then the quadratic form $\mathfrak{q}(x)$ is odd and $Q_m(a)=-2^{m_1}$.
In fact, the parity of $\mathfrak{q}(x)$
can be determined by the point multiplication
of a random element in $J(C)(\mathbb{F}_{2^m})$.
Consequently, the sign of $Q_{m}(a)$ is
determined \cite{CNa}.
\end{remark}

\section{A class of hyper-bent functions
in binomial forms}
\label{sec:3}
\subsection{Boolean functions in $\mathcal{H}_{n}$}
\label{subsec:3.1}

Through the paper, we assume that
$n=2m$ and $m\equiv 2\pmod 4$.
Let $\mathcal{H}_{n}$ be the set of Boolean
functions on $\mathbb{F}_{2^n}$ of the form:
\begin{equation}\label{equation1}
f_{a,b}(x):=\mathrm{Tr}_{1}^{n}(ax^{2^m-1})
+\mathrm{Tr}_{1}^{4}(bx^{\frac{2^n-1}{5}}),
\end{equation}
where $a\in \mathbb{F}_{2^m}$ and $b\in \mathbb{F}_{16}$.

Note that the cyclotomic coset of 2 module $2^n-1$  containing $\frac{2^n-1}{5}$ is $\{\frac{2^n-1}{5}, 2\cdot
\frac{2^n-1}{5}, 2^2\cdot \frac{2^n-1}{5},
2^3\cdot \frac{2^n-1}{5}
\}$. Then its size
is 4, that is, $o(\frac{2^n-1}{5})=4$. Hence, the Boolean function $f_{a,b}$ is
not in the class considered by Charpin and Gong [6]. Further, it does not lie in
the class of Boolean functions studied by
Mesnager \cite{SM,SMT}.

Since $m\equiv 2\pmod 4$ and $2^m+1\equiv 0
\pmod 5$, any Boolean function $f_{a,b}$ in $\mathcal{H}_{n}$ satisfies
$$
f_{a,b}(\alpha^{2^m+1}x)=f_{a,b}(x), \forall x\in
\mathbb{F}_{2^n},
$$
where $\alpha$ is a primitive element
in $\mathbb{F}_{2^n}$. Note that $f_{a,b}(0)
=0$. Then the hyper-bentness $f_{a,b}$ can
be characterized by the following proposition.
\begin{proposition}\label{prop3.1}
Let $f_{a,b}\in\mathcal{H}_{n}$.  Set the following sum:
\begin{equation}\label{equation2}
\Lambda(a,b):=\sum\limits_{u\in
U}\chi(f_{a,b}(u))
\end{equation}
where $U$ is the group of all $2^m+1$-th
roots of unity in $\mathbb{F}_{2^n}$, that is,
$U=\{x\in\mathbb{F}_{2^n}|x^{2^m+1}=1\}$. Then $f_{a,b}$ is a hyper-bent function if and only if $\Lambda(a,b)=1$. Further, a hyper-bent function $f_{a,b}$ lies in $\mathcal{PS}_{ap}$ if and only if $\mathrm{Tr}_{1}^{4}(b)=0$.
\end{proposition}
\begin{proof}
From Proposition \ref{prop2.3},
$f_{a,b}$ is a hyper-bent function if and only if
its restriction to $U$ has Hamming weight
$2^{m-1}$. Note that
\begin{align*}
\Lambda(a,b)&=\sum_{x\in U}\chi(f_{a,b}(u))\\
&=\#\{u\in U|f_{a,b}(u)=0\}-
\#\{u|f_{a,b}(u)=1\}\\
&=\#U-2\#\{u|f_{a,b}(u)=1\}\\
&=2^m+1-2\#\{u|f_{a,b}(u)=1\}.
\end{align*}
Then the restriction of $f_{a,b}$ to
$U$ has Hamming weight $2^{m-1}$ if
and only if $\Lambda(f_{a,b})=1$.
Hence, $f_{a,b}$ is a hyper-bent if and only if
$\Lambda(f_{a,b})=1$.
As for the second part of this proposition,
we get that
\begin{align*}
f_{a,b}(1)=&\mathrm{Tr}_{1}^{n}(a)
+\mathrm{Tr}_{1}^{4}(b)\\
=&\mathrm{Tr}_{1}^{m}(a+a^{2^m})+
\mathrm{Tr}_{1}^{4}(b)\\
=&\mathrm{Tr}_{1}^{4}(b).
\end{align*}
Then $f_{a,b}(1)=0$ if and only if
$\mathrm{Tr}_{1}^{4}(b)=0$. Hence, from Proposition \ref{prop2.3},  $f_{a,b}$ lies in $\mathcal{PS}_{ap}$ if and only if $\mathrm{Tr}_{1}^{4}(b)=0$.
\end{proof}

\subsection{Character sums on Boolean functions in $\mathcal{H}_{n}$}
\label{subsec:3.2}

Dillon \cite{JD} presented the characterization of the hyper-bentness of $f_{a,0}$. In this paper,
we consider the hyper-bentness of
$f_{a,b}(b\neq 0)$ in $\mathcal{H}_{n}$.
We first give some notations and present
some properties of character sums.

Let $\alpha$ be a primitive element in
$\mathcal{H}_{n}$. Then $\beta=\alpha^{\frac{2^n-1}{5}}$ is a primitive
5-th root of unity in $U$, where
$U$ is the cyclic group generated by
$\xi=\alpha^{2^m-1}$. Let
$V$  be the cyclic group generated by
$\alpha^{5(2^m-1)}$. Then  we have
$$
U=\cup_{i=0}^{4}\xi^iV,\quad \mathbb{F}_{2^{n}}^{*}
=\mathbb{F}_{2^m}^{*}\times U.
$$
We introduce the character sums:
$$
S_{i}=\sum_{v\in V}\chi(\mathrm{Tr}_{1}^{n}(
a(\xi^iv)^{2^m-1})).
$$
Then, we have
\begin{equation}\label{equation3}
S_{0}+S_{1}+S_{2}
+S_{3}+S_{4}=\sum_{u\in U}\chi(\mathrm{Tr}_{1}^{n}(au^{2^m-1}))
=\Lambda(a,0).
\end{equation}
Obviously, for any integer
$i$, $S_{i}
=S_{i\pmod 5}$. Further, we have the following
lemma on $S_{i}$.
\begin{lemma}\label{lem3.1}
$S_{1}=S_{4}$, $S_{2}=S_{3}$.
\end{lemma}
\begin{proof}
Note that $\mathrm{Tr}_{1}^{n}
(x^{2^m})=\mathrm{Tr}_{1}^{n}(x)$, then
\begin{align*}
S_{i}=\sum_{v\in V}\chi(\mathrm{Tr}_{1}^{n}(a
(\xi^i v)^{(2^m-1)}))
=\sum_{v\in V}\chi(
\mathrm{Tr}_{1}^{n}(a^{2^m}
(\xi^{i2^m}v^{2^{m}})^{(2^m-1)})).
\end{align*}
From $a\in \mathbb{F}_{2^m}$, $a^{2^m}=a$.
Since $m\equiv 2\pmod 4$ and $2^m\equiv -1\pmod 5$, hence $i2^m\equiv -i \pmod 5$ and $\xi^{i2^m}
v^{2^m}=\xi^{-i}(\xi^{i(2^m+1)}v^{2^m})$, where $
\xi^{i(2^m+1)}\in V$. The map
$
v\longmapsto \xi^{i(2^m+1)}v^{2^m}
$
is a permutation of $V$. Consequently,
$$
S_{i}=\sum_{v\in V}\chi(\mathrm{Tr}_{1}^{n}
(a(\xi^{-i} v)^{(2^m-1)}))=S_{-i}.
$$
We just take $i=1,2$. Then this lemma follows.
\end{proof}
From (\ref{equation3}) and Lemma \ref{lem3.1}, we
can get the following corollary.
\begin{corollary}\label{cor3.1}
$S_{0}+2(S_{1}+S_{2})=\Lambda(a,0)$.
\end{corollary}
Further, $\Lambda(a,b)$ is a linear combination of
$S_{0}$, $S_{1}$ and $S_{2}$.
\begin{proposition}\label{prop3.2}
$\Lambda(a,b)=\chi(\mathrm{Tr}_{1}^{4}(b))S_{0}
+(\chi(\mathrm{Tr}_{1}^{4}(b\beta^2))
+\chi(\mathrm{Tr}_{1}^{4}(b\beta^3)))S_{1}
+(\chi(\mathrm{Tr}_{1}^{4}(b\beta))
+\chi(\mathrm{Tr}_{1}^{4}(b\beta^4)))S_{2}$.
\end{proposition}
\begin{proof}
From $(\ref{equation2})$, we have
\begin{align*}
\Lambda(a,b)=&\sum_{u\in U}\chi(f_{a,0}(u)
+\mathrm{Tr}_{1}^{4}(bu^{\frac{2^n-1}{5}}))\\
=& \sum_{u\in U}\chi(\mathrm{Tr}_{1}^{4}(bu^{
\frac{2^n-1}{5}}))\chi(f_{a,0}(u))\\
=& \sum_{i=0}^{4}\sum_{v\in V}
\chi(\mathrm{Tr}_{1}^{4}(b(\xi^i v)^{\frac{2^n-1}{5}}))\chi(f_{a,0}(\xi^i v))
\quad (\text{From $(\ref{equation3})$})\\
=& \sum_{i=0}^{4}\sum_{v\in V}
\chi(\mathrm{Tr}_{1}^{4}(b(\alpha^{i(2^{m}-1)})^{\frac{2^n-1}{5}}))\chi(f_{a,0}(\xi^i v))\quad (\xi=\alpha^{2^m-1})\\
=& \sum_{i=0}^{4}\sum_{v\in V}
\chi(\mathrm{Tr}_{1}^{4}(b\beta^{i(2^m-1)}))
\chi(f_{a,0}(\xi^i v))\quad (\beta=\alpha^{
\frac{2^n-1}{5}}).
\end{align*}
Since $2^m+1\equiv 0\pmod 5$, $2^m-1\equiv 3\pmod 5$ and
\begin{align*}
\Lambda(a,b)= \sum_{i=0}^{4}\sum_{v\in V}
\chi(\mathrm{Tr}_{1}^{4}(b\beta^{3i}))
\chi(f_{a,0}(\xi^i v))
=\sum_{i=0}^{4}\chi(\mathrm{Tr}_{1}^{4}
(b\beta^{3i}))\sum_{v\in V}\chi(f_{a,0}(\xi^i v)).
\end{align*}
From the definition of $S_{i}$,
\begin{equation*}
\Lambda(a,b)=\sum_{i=0}^{4}\chi(\mathrm{Tr}_{1}^{4}
(b\beta^{3i}))S_{i}.
\end{equation*}
From Lemma $\ref{lem3.1}$, this proposition
follows.
\end{proof}
Further, if $a\in\mathbb{F}_{2^{m_{1}}}$, where $m_{1}=m/2$, we
have the following proposition.
\begin{proposition}\label{prop3.3}
If $a\in \mathbb{F}_{2^{m_{1}}}$, where  $m_{1}=m/2$, then
$$
S_{1}=S_{2}, S_{0}+4S_{1}=\Lambda(a,0).
$$
\end{proposition}
\begin{proof}
From $a\in \mathbb{F}_{2^{m_{1}}}$,
$\mathrm{Tr}_{1}^{n}(ax^{(2^m-1)})
=\mathrm{Tr}_{1}^{n}(ax^{2^{m_{1}}
(2^m-1)})$. Then
\begin{align*}
S_{i}=\sum_{v\in V}\chi(
\mathrm{Tr}_{1}^{n}(a(\xi^i v)^{(2^m-1)}))
=\sum_{v\in V}\chi(
\mathrm{Tr}_{1}^{n}(a(\xi^{2^{m_{1}}i} v^{2^{m_{1}}})^{(2^m-1)})).
\end{align*}
We take $i=1$. Then
$$
S_{1}=\sum_{v\in V}\chi(
\mathrm{Tr}_{1}^{n}(a(\xi^{2^{m_{1}}} v^{2^{m_{1}}})^{(2^m-1)})).
$$

Since $2^m+1\equiv 0\pmod 5$, $(2^{m_{1}})^2\equiv -1\pmod 5$ and $2^{m_{1}}\equiv \pm 2\pmod 5$.

When $2^{m_{1}}\equiv 2\pmod 5$,
$$
S_{1}=\sum_{v\in V}\chi(
\mathrm{Tr}_{1}^{n}(a(\xi^{2}\xi^{2^{m_{1}}-2} v^{2^{m_{1}}})^{(2^m-1)})).
$$
The map $
v\longmapsto \xi^{2^{m_{1}}-2}v^{2^{m_{1}}}
$ is a permutation of $V$. Consequently,
$$
S_{1}=\sum_{v\in V}\chi(
\mathrm{Tr}_{1}^{n}(a(\xi^{2}v)^{(2^m-1)}))
=S_{2}.
$$

When $2^{m_{1}}\equiv -2\pmod 5$, we can similarly have $S_{1}=S_{3}$.

Above all, $S_{1}=S_{2}$. From Corollary \ref{cor3.1}, $S_{0}+4S_{1}=\Lambda(a,0)$.
\end{proof}
For $\Lambda(a,b)$, the following proposition
gives some properties.
\begin{proposition}\label{prop3.4}
$\Lambda(a,b)$ satisfies

{\rm (1)} $\Lambda(a,b^4)=\Lambda(a,b)$.

{\rm (2)} If $b$ is a primitive element of $\mathbb{F}_{2^{16}}$ and $\mathrm{Tr}_{1}^{4}(b)=0$,
$\Lambda(a,b^2)=\Lambda(a,b)=S_{0}$.
\end{proposition}
\begin{proof}
{\rm (1)}  For any $b\in \mathbb{F}_{16}$, $
\mathrm{Tr}_{1}^{4}(b^4)=\mathrm{Tr}_{1}^{4}(b)$.
Then
\begin{align*}
\mathrm{Tr}_{1}^{4}(b(\beta^2+\beta^3))
=\mathrm{Tr}_{1}^{4}(b^4(\beta^{8}+\beta^{12}))
=\mathrm{Tr}_{1}^{4}(b^4(\beta^2+\beta^3))\\
\mathrm{Tr}_{1}^{4}(b(\beta+\beta^4))
=\mathrm{Tr}_{1}^{4}(b^4(\beta^4+\beta^{16}))
=\mathrm{Tr}_{1}^{4}(b^4(\beta+\beta^4)).
\end{align*}
From Proposition \ref{prop3.2}, $\Lambda(a,b^4)=\Lambda(a,b)$.

{\rm (2)} If $b$ is a primitive element in
$\mathbb{F}_{16}$ such that $\mathrm{Tr}_{1}^{4}(b)=0$, then
$b^4+b+1=0$. Further,
\begin{align*}
\mathrm{Tr}_{1}^{4}(b(\beta^2+\beta^3))=&
\mathrm{Tr}_{1}^{2}(b^4(\beta^2+\beta^3)
+b(\beta^2+\beta^3))\\
=&\mathrm{Tr}_{1}^{2}((b+b^4)(\beta^2+\beta^3))\\
=&\mathrm{Tr}_{1}^{2}(\beta^2+\beta^3).
\end{align*}
The minimal polynomial of $\beta$ in $\mathbb{F}_{2}$ is $\beta^4+\beta^3+\beta^2+\beta+1=0$. Then
$$
\mathrm{Tr}_{1}^{2}(\beta^2+\beta^3)
=\beta^2+\beta^3+\beta^4+\beta^6=1.
$$
Consequently, $
\mathrm{Tr}_{1}^{4}(b(\beta^2+\beta^3))=1$.
Similarly, $\mathrm{Tr}_{1}^{4}(b(\beta+\beta^4))=1$.
Therefore,
$$
\chi(\mathrm{Tr}_{1}^{4}(b\beta^2))
+\chi(\mathrm{Tr}_{1}^{4}(b\beta^3))=0, \quad
\chi(\mathrm{Tr}_{1}^{4}(b\beta))+
\chi(\mathrm{Tr}_{1}^{4}(b\beta^4))=0.
$$
From Proposition \ref{prop3.2}, $\Lambda(a,b)=S_{0}$.

Since $b$ is a primitive element such that
$\mathrm{Tr}_{1}^{4}(b)=0$, $b^2$ is also a
primitive element such that $\mathrm{Tr}_{1}^{4}(b^2)=0$. Consequently,
$\Lambda(a,b^2)
=\Lambda(a,b)=S_{0}$.
\end{proof}

Explicit results on $\Lambda(a,b)$ are given in
the following proposition.
\begin{proposition}\label{prop3.5}
Let $b\in\mathbb{F}_{16}^{*}$, then

{\rm (1)} If $b=1$,  $\Lambda(a,b)=S_{0}-2(S_{1}+S_{2})=2S_{0}-
\Lambda(a,0)$.

{\rm (2)} If $b\in\{\beta+\beta^2,
\beta+\beta^3, \beta^2+\beta^4,\beta^3+
\beta^4\}$, that is, $b$ is a primitive element satisfies $\mathrm{Tr}_{1}^{4}
(b)=0$, then $\Lambda(a,b)
=S_{0}$.

{\rm (3)} If $b=\beta$ or $\beta^4$, then $\Lambda(a,b)=-S_{0}-2S_{1}$.

{\rm (4)} If $b=\beta^2$ or $\beta^3$, then $\Lambda(a,b)=-S_{0}-2S_{2}$.

{\rm (5)} If $b=1+\beta$ or $1+\beta^4$, then
$\Lambda(a,b)=-S_{0}+2S_{1}$.

{\rm (6)} If $b=1+\beta^2$ or $1+\beta^3$, $\Lambda(a,b)=-S_{0}+2S_{2}$.

{\rm (7)} If $b=\beta+\beta^4$, then
$\Lambda(a,b)=S_{0}+2S_{1}-2S_{2}$.

{\rm (8)} If $b=\beta^2+\beta^3$, then $\Lambda(a,b)=S_{0}-2S_{1}+2S_{2}$.
\end{proposition}
\begin{proof}
From Proposition \ref{prop3.2}, this proposition
follows.
\end{proof}
Further, if $a\in\mathbb{F}_{2^{m_{1}}}$,
we have the following proposition.
\begin{proposition}\label{prop3.6}
If $a\in\mathbb{F}_{2^{m_{1}}}$, then

{\rm (1)} If $b=1$, then $\Lambda(a,b)=
2S_{0}-\Lambda(a,0)$.

{\rm (2)} If $b\in\{\beta,\beta^2,\beta^3,
\beta^4\}$, then $\Lambda(a,b)=-S_{0}-2S_{1}
=-\frac{S_{0}+\Lambda(a,0)}{2}$.

{\rm (3)} If $b\in\{1+\beta,1+\beta^2,1+\beta^3,
1+\beta^4\}$, then $\Lambda(a,b)
=-S_{0}+2S_{1}=-\frac{3S_{0}-\Lambda(a,0)}{2}$.

{\rm (4)} If $b\in\{
\beta+\beta^2, \beta+\beta^3, \beta^2+\beta^4,
\beta^3+\beta^4,
\beta+\beta^4,\beta^2+\beta^3
\}$, then $\Lambda(a,b)=S_{0}$.
\end{proposition}
\begin{proof}
From Proposition \ref{prop3.3} and
Proposition \ref{prop3.5}, this proposition follows.
\end{proof}
\begin{corollary}\label{cor3.2}
For $a\in\mathbb{F}_{2^{m_{1}}}$,  $\Lambda(a,b^2)=\Lambda(a,b)$.
\end{corollary}
\begin{proof}
This corollary can be obtained by Proposition \ref{prop3.6}.
\end{proof}

We now introduce some results on character sums.
\begin{lemma}\label{lem3.2}
Let $U$ be the group of $2^m+1$-th roots of unity in $\mathbb{F}_{2^n}^{*}$, then for any positive integer $p$,
$$
\sum_{u\in U}\chi(\mathrm{Tr}(ax^{p(2^m-1)}))
=1+2\sum_{x\in \mathbb{F}_{2^m}^{*},
\mathrm{Tr}_{1}^{m}(x^{-1})=1}
\chi(\mathrm{Tr}(aD_{p}(x))).
$$
\end{lemma}
\begin{proof}
This lemma is a special case of
Lemma 12 by Mesnager \cite{SMT}.
\end{proof}
\begin{proposition}\label{prop3.7}
Let $S_{0}$ and $\Lambda(a,0)$ be defined
above, then

{\rm (1)} $\Lambda(a,0)=1-K_{m}(a)$.

{\rm (2)} $S_{0}=\frac{1}{5}[1-K_{m}(a)
+2Q_{m}(a)]$.
\end{proposition}
\begin{proof}
{\rm (1)} From Lemma \ref{lem3.2},
\begin{align*}
\Lambda(a,0)=&\sum_{u
\in U}\chi(\mathrm{Tr}_{1}^{n}(au^{2^m-1}))\\
=&1+2\sum_{x\in\mathbb{F}_{2^m}^{*
},\mathrm{Tr}_{1}^{m}(x^{-1})=1}
\chi(\mathrm{Tr}_{1}^{m}(ax))\\
=&1+2\cdot\frac{1}{2}[\sum_{x\in\mathbb{F}_{2^m}}
\chi(\mathrm{Tr}_{1}^{m}(ax))-
\sum_{x\in\mathbb{F}_{2^m}}
\chi(\mathrm{Tr}_{1}^{m}(ax+\frac{1}{x}))].
\end{align*}
Note that $\sum\limits_{x\in\mathbb{F}_{2^m}}
\chi(\mathrm{Tr}_{1}^{m}(ax))=0$, then
$$
\Lambda(a,0)=1-\sum_{x\in\mathbb{F}_{2^m}}
\chi(\mathrm{Tr}_{1}^{m}(ax+\frac{1}{x}))
=1-K_{m}(a).
$$

{\rm (2)} Note that $
S_{0}= \sum\limits_{v\in V}^{}\chi(\mathrm{Tr}_{1}^{n}
(av^{2^m-1}))
=\frac{1}{5}\sum\limits_{u\in U}\chi
(\mathrm{Tr}_{1}^{n}
(au^{5(2^m-1})))$. Then from
Lemma \ref{lem3.2},
\begin{align*}
S_{0}=&\frac{1}{5}[1+2
\sum_{x\in\mathbb{F}_{2^m}^{*
},\mathrm{Tr}_{1}^{m}(x^{-1})=1}
\chi(\mathrm{Tr}_{1}^{m}(aD_{5}(x)))
]\\
=&
\frac{1}{5}[1+2
\sum_{x\in\mathbb{F}_{2^m}}
\chi(\mathrm{Tr}_{1}^{m}(aD_{5}(x)))
-2
\sum_{x\in\mathbb{F}_{2^m},
\mathrm{Tr}_{1}^{m}(x^{-1})=0}
\chi(\mathrm{Tr}_{1}^{m}(aD_{5}(x)))
].
\end{align*}
Since $
\frac{1}{D_{5}(x)}=\frac{1}{x^5+x^3+x}
=\frac{1}{x}+\frac{x}{x^2+x+1}
+(\frac{x}{x^2+x+1})^2$,
$
\mathrm{Tr}_{1}^{m}(\frac{1}{D_{5}(x)})
=\mathrm{Tr}_{1}^{m}(\frac{1}{x})$.
Then $D_{5}(x)$ induces the map
\begin{eqnarray*}
\{x\in\mathbb{F}_{2^m}|\mathrm{Tr}_{1}^{m}(x^{-1})
=0\}&\longrightarrow& \{x\in\mathbb{F}_{2^m}|\mathrm{Tr}_{1}^{m}(x^{-1})
=0\}\\
x&\longmapsto & D_{5}(x).
\end{eqnarray*}
From Proposition \ref{prop2.4},
this map is a permutation of $\{x\in\mathbb{F}_{2^m}|\mathrm{Tr}_{1}^{m}(x^{-1})
=0\}$. Hence,
\begin{align*}
&\sum_{x\in\mathbb{F}_{2^m},
\mathrm{Tr}_{1}^{m}(x^{-1})=0}
\chi(\mathrm{Tr}_{1}^{m}(aD_{5}(x)))\\
=& \sum_{x\in\mathbb{F}_{2^m},
\mathrm{Tr}_{1}^{m}(x^{-1})=0}
\chi(\mathrm{Tr}_{1}^{m}(ax))\\
=& \frac{1}{2} [
\sum_{x\in\mathbb{F}_{2^m}}
\chi(\mathrm{Tr}_{1}^{m}(ax))+
\sum_{x\in\mathbb{F}_{2^m}}
\chi(\mathrm{Tr}_{1}^{m}(ax+\frac{1}{x}))
]\\
=& \frac{1}{2}K_{m}(a).
\end{align*}
Consequently, $S_{0}=\frac{1}{5}[1+2Q_{m}(a)-K_{m}(a)]$.
\end{proof}
The relation between $K_{m}(a)$,
$Q_{m}(a)$ and $\Lambda(a,b)$ is given in the following proposition.
\begin{proposition}\label{prop3.8}
Let $\Lambda(a,b)$ be defined above, where
~$a\in\mathbb{F}_{2^m}$, $b\in\mathbb{F}_{16}^{*}$, then

{\rm (1)} If $b=1$, then
$\Lambda(a,1)=-\frac{1}{5}[3(1-K_{m}(a))-
4Q_{m}(a)]$.

{\rm (2)} If $b$ is a primitive element such that $\mathrm{Tr}_{1}^{4}(b)
=0$, then $\Lambda(a,b)=\frac{1}{5}[1-K_{m}(a)+
2Q_{m}(a)]$.
\end{proposition}
\begin{proof}
From Proposition \ref{prop3.5}, $\Lambda(a,1)=2S_{0}-\Lambda(a,0)$. Further,
Proposition \ref{prop3.7},
\begin{align*}
\Lambda(a,1)=&2\cdot\frac{1}{5}[1-K_{m}(a)+2Q_{m}(a)]
-(1-K_{m}(a))\\
=& -\frac{3}{5}(1-K_{m}(a))+\frac{4}{5}Q_{m}(a)\\
=& -\frac{1}{5}[3(1-K_{m}(a))-4Q_{m}(a)]
\end{align*}

{\rm (2)} From Proposition \ref{prop3.5} and
Proposition \ref{prop3.7}, $\Lambda(a,b)=S_{0}=\frac{1}{5}[1-K_{m}(a)+
2Q_{m}(a)]$.
\end{proof}

\subsection{The hyper-bentness of
$f_{a,b} (a \in\mathbb{F}_{2^m},(b+1)(b^4+b+1)=0)$}
\label{subsec:3.3}

\begin{proposition}\label{prop3.9}
Let $n=2m$, $m=2m_{1}$, $m_{1}\equiv 1\pmod 2$  and $m_{1}\geq 3$, the Boolean function $
f_{a,1}$ in $\mathcal{H}_{n}$ of the form
$$
\mathrm{Tr}_{1}^{n}(ax^{2^m-1})+
\mathrm{Tr}_{1}^{4}(x^{\frac{2^n-1}{5}})
$$
is a hyper-bent function if and only if $Q_{m}(a)=2^{m_{1}}$ and
~$K_{m}(a)=\frac{4}{3}(2-2^{m_{1}})$ holds.
\end{proposition}
\begin{proof}
From Proposition \ref{prop3.1}, $f_{a,1}$ is a hyper-bent function if and only if $\Lambda(a,1)=1$. From Proposition \ref{prop3.8},
$
\Lambda(a,1)=-\frac{1}{5}[3(1-K_{m}(a))
-4Q_{m}(a)]$. Then
$\Lambda(a,1)=1$ if and only if $K_{m}(a)
=\frac{4}{3}(2-Q_{m}(a))$.
Note that
$$
1-2\cdot 2^{m_{1}}\leq K_{m}(a)\leq
1+2\cdot 2^{m_{1}}.
$$
Further, we have
$$
-\frac{3}{2}\cdot 2^{m_{1}}+\frac{5}{4}\leq Q_{m}(a)\leq
\frac{3}{2}\cdot 2^{m_{1}}+\frac{5}{4}.
$$
From $m_1\geq 3$,
$$-2\cdot
2^{m_{1}}\leq Q_{m}(a)
\leq 2\cdot 2^{m_{1}}.$$
From Corollary \ref{cor2.1},
$Q_{m}(a)
\in \{0,\pm 2^{m_{1}},
\pm 2\cdot 2^{m_{1}},
\pm 4\cdot 2^{m_{1}}\}$. Hence, the value of $Q_{m}(a)$ is $0$ or $\pm2^{m_{1}}$.

If $Q_{m}(a)=0$, then $K_{m}(a)=
\frac{4}{3}(2-0)=\frac{8}{3}$. Since
$K_{m}(a)$ is an integer, $Q_{m}(a)\neq 0$.

If $Q_{m}(a)=-2^{m_{1}}$, then $K_{m}(a)=\frac{4}{3}(2+2^{m_{1}})
=\frac{8}{3}(1+2^{m_{1}-1})$. Since $m_{1}$ is odd, $3\not|(1+2^{m_{1}-1})$. Hence, $Q_{m}(a)\neq -2^{m_{1}}$.

As a result, $f_{a,1}$ is a hyper-bent function if and only if $Q_{m}(a)
=2^{m_{1}}$. Then $K_{m}(a)
=\frac{4}{3}(2-2^{m_{1}})$.  Hence,
this proposition follows.
\end{proof}

\begin{proposition}\label{prop3.10}
Let $n=2m$, $m=2m_{1}$,
$m_{1}\equiv 1\pmod 2$ and
$m_{1}\geq 3$. Let
$b$ be a primitive element
in $\mathbb{F}_{16}^{*}$ such that $\mathrm{Tr}_{1}^{4}(b)=0$.
The Boolean function $f_{a,b}$ in
$\mathcal{H}_{n}$ of the form
$$
\mathrm{Tr}_{1}^{n}(ax^{2^m-1})+
\mathrm{Tr}_{1}^{4}(bx^{\frac{2^n-1}{5}})
$$
is a hyper-bent function if and only if one of the following assertions $(1)$ and $(2)$ holds.

{\rm (1)} $Q_{m}(a)=0$, $K_{m}(a)=-4$.

{\rm (2)} $Q_{m}(a)=2^{m_{1}}$, $K_{m}(a)=2\cdot2^{m_{1}}-4$.
\end{proposition}
\begin{proof}
$f_{a,b}$ is a hyper-bent function if and only if $\Lambda(a,b)=1$. From Proposition
\ref{prop3.5}, when $b$ is a primitive element such that $\mathrm{Tr}_{1}^{4}(b)=0$,
$\Lambda(a,b)=S_{0}$.
From Proposition \ref{prop3.7},
$
\Lambda(a,b)=\frac{1}{5}
[1-K_{m}(a)+2Q_{m}(a)]$.
Hence, $\Lambda(a,b)=1$ if and only if $K_{m}(a)=2Q_{m}(a)-4$.

From Corollary \ref{cor2.1},
$Q_{m}(a)\in\{0,\pm2^{m_{1}},\pm2\cdot 2^{m_{1}},
\pm 4\cdot 2^{m_{1}}
\}$.
Further, from Proposition \ref{prop2.6},
$K_{m}(a)\in[1-2\cdot2^{m_{1}}, 1+
2\cdot2^{m_{1}}]$. Hence,
$Q_{m}(a)=0$ or $2^{m_{1}}$.
If $Q_{m}(a)=0$, then $K_{m}(a)=-4$.
If $Q_{m}(a)=2^{m_{1}}$, then $K_{m}(a)
=2\cdot2^{m_{1}}-4$. Therefore,
this proposition follows.
\end{proof}

\begin{theorem}\label{thm3.1}
Let $n=2m$, $m=2m_{1}$,
$m_{1}\equiv 1\pmod 2$ and $m_{1}\geq 3$.
The Boolean function $f_{a,1}$ in
$\mathcal{H}_{n}$ of the form
$$
\mathrm{Tr}_{1}^{n}(ax^{2^m-1})+
\mathrm{Tr}_{1}^{4}(x^{\frac{2^n-1}{5}})
$$
is a hyper-bent function if and only if the following assertions holds.

{\rm (1)} $p(x)=x^5+x+a^{-1}$ is irreducible over  $\mathbb{F}_{2^m}$.

{\rm (2)} The quadratic form $\mathfrak{q}(x)=
\mathrm{Tr}_{1}^{m}(x(ax^4+ax^2+a^2x))$ over
$\mathbb{F}_{2^m}$ is even.

{\rm (3)} $K_{m}(a)=\frac{4}{3}(2-2^{m_{1}})$.
\end{theorem}
\begin{proof}
From Proposition \ref{prop3.9} and
Corollary \ref{cor2.2}, this theorem follows.
\end{proof}
\begin{theorem}\label{thm3.2}
Let $n=2m$, $m=2m_{1}$,
$m_{1}\equiv 1\pmod 2$ and
$m_{1}\geq 3$.
Let $b$ be  a primitive element in $\mathbb{F}_{16}^{*}$ such that $\mathrm{Tr}_{1}^{4}(b)=0$.
The Boolean function $f_{a,b}$ in
$\mathcal{H}_{n}$ of the form
$$
\mathrm{Tr}_{1}^{n}(ax^{2^m-1})+
\mathrm{Tr}_{1}^{4}(bx^{\frac{2^n-1}{5}})
$$
is a hyper-bent function if and only if one of the assertions $(1)$ and $(2)$ holds.

{\rm (1)} $p(x)=x^5+x+a^{-1}$ over $\mathbb{F}_{2^m}$ is $(1)(2)^2$ and $K_{m}(a)=-4$.

{\rm (2)}
$p(x)=x^5+x+a^{-1}$ is irreducible over $\mathbb{F}_{2^m}$.  The quadratic form $\mathfrak{q}(x)=
\mathrm{Tr}_{1}^{m}(x(ax^4+ax^2+a^2x))$ over $\mathbb{F}_{2^m}$ is even. $K_{m}(a)=2\cdot2^{m_{1}}-4$.
\end{theorem}
\begin{proof}
From \ref{prop3.10} and Corollary \ref{cor2.2},
this theorem follows.
\end{proof}

\subsection{The hyper-bentness of $f_{a,b}~(a\in\mathbb{F}_{2^{\frac{m}{2}}}, b\in\mathbb{F}_{16})$}

To consider the hyper-bentness of $f_{a,b}~(a\in\mathbb{F}_{2^{\frac{m}{2}}}, b\in\mathbb{F}_{16})$, we require more properties of $K_{m}(a)$ and $Q_{m}(a)$.
\begin{lemma}\label{lem3.3}
Let $a\in \mathbb{F}_{2^{m_{1}}}^{*}$, $m=2m_{1}$ and $p(x)=x^5+x+a^{-1}$, then

{\rm (1)} $1-K_{m}(a)=(1-K_{m_{1}}(a))^2-2\cdot 2^{m_{1}}$.

{\rm (2)} If $m_{1}\equiv 1 \pmod 2$, then
$Q_{m}(a)\in \{0, 2\cdot 2^{m_{1}},
-4\cdot 2^{m_{1}}\}$. Further, $Q_{m}(a)=0$ if and only if $p(x)=(1)(4)$.  $Q_{m}(a)=2\cdot 2^{m_{1}}$ if and only if  $p(x)=(2)(3)$.
$Q_{m}(a)=-4\cdot 2^{m_{1}}$ if and only if $p(x)=(1)^3(2)$.
\end{lemma}
\begin{proof}
{\rm (1)} We introduce an auxiliary curve
$$
C:y^2+y=ax+\frac{1}{x}, a\in
\mathbb{F}_{2^{m_{1}}},
$$
that is, $xy^2+xy=ax^2+1$. Hence, $C$ is an
elliptic curve. There are two infinite points
on $C$ and the $x$ coordinate of any point
on $C$ is not zero. Therefore,
\begin{align*}
\#(C(\mathbb{F}_{2^{m_{1}}}))=&
2+2\cdot \#\{x\in\mathbb{F}_{2^{m_{1}}}^*|
\mathrm{Tr}_{1}^{m_{1}}(ax+\frac{1}{x})=0
\}\\
=& 2\cdot \#\{x\in\mathbb{F}_{2^{m_{1}}}|
\mathrm{Tr}_{1}^{m_{1}}(ax+\frac{1}{x})=0
\}\\
=& \cdot \#\{x\in\mathbb{F}_{2^{m_{1}}}|
\mathrm{Tr}_{1}^{m_{1}}(ax+\frac{1}{x})=0
\}-2^{m_{1}}+2^{m_{1}}\\
=& \sum_{x\in\mathbb{F}_{2^{m_{1}}}}
\chi(\mathrm{Tr}_{1}^{m_{1}}(ax+\frac{1}{x}))
+2^{m_{1}}\\
=& (1+2^{m_{1}})+K_{m_{1}}(a)-1.
\end{align*}
Further, $
\#(C(\mathbb{F}_{2^{2m_{1}}}))=
(1+2^{2m_{1}})+K_{2m_{1}}(a)-1$.
From properties of elliptic curves, we have
$$
1-K_{m}(a)=1-K_{2m_{1}}(a)=(1-K_{m_{1}})^2-2\cdot 2^{m_{1}}.
$$

{\rm (2)} Consider $p(x)=x^5+x+a^{-1}$ over
$\mathbb{F}_{2^{m_{1}}}$. From Proposition
\ref{prop2.8} and (\ref{equation}),
$$Q_{m}(a)=2s-r^2,$$
where $(r,s)$ is determined by $m_{1}$ and Proposition \ref{prop2.7}.
Therefore, $Q_{m}(a)=0$ if and only if $p(x)=(1)(4)$.  $Q_{m}(a)=2\cdot 2^{m_{1}}$ if and only if  $p(x)=(2)(3)$.
$Q_{m}(a)=-4\cdot 2^{m_{1}}$ if and only if $p(x)=(1)^3(2)$.
\end{proof}
\begin{corollary}\label{cor3.3}
Let $a\in\mathbb{F}_{2^{m_{1}}}$,
$m=2m_{1}$ and  $m_{1}\equiv 1\pmod2$, then $K_{m}(a)\neq -4$.
\end{corollary}
\begin{proof}
From Lemma \ref{lem3.3},  if $K_{m}(a)= -4$,
\begin{equation}\label{equ1}
(1-K_{m_{1}})^2=2\cdot 2^{m_{1}}+5.
\end{equation}
Since $m_{1}\geq 3$ and $m_{1}\equiv 1\pmod2$,
$(2^{\frac{m_{1}+1}{2}})^2<
(1-K_{m_{1}}(a))^2< (2^{\frac{m_{1}+1}{2}}+1)^2$.
Then $(\ref{equ1})$ has no integer solution.
Hence $K_{m}(a)\neq -4$.
\end{proof}
\begin{theorem}\label{thm3.3}
Let $n=2m$, $m=2m_{1}$, $m_{1}\equiv 1\pmod 2$ and $m_{1}\geq 3$.
If $b\in\mathbb{F}_{16}\backslash \{\beta^i|
1\leq i\leq 4\}$, then the
Boolean function $f_{a,b}$ in
$\mathcal{H}_{n}$ of the form
$$
\mathrm{Tr}_{1}^{n}(ax^{2^m-1})+
\mathrm{Tr}_{1}^{4}(bx^{\frac{2^n-1}{5}}),
a\in \mathbb{F}_{2^{m_{1}}}
$$
is not a hyper-bent function.
\end{theorem}
\begin{proof}
When $b=0$, $f_{a,0}$ is a hyper-bent function if and only if $\Lambda(a,0)=1$.
From Proposition \ref{prop3.7},
$\Lambda_{a,0}=1-K_{m}(a)$. Hence,
$f_{a,0}$ is a hyper-bent function if and only if $K_{m}(a)=0$.
From Lemma \ref{lem3.3}, $(1-K_{m}(a))^2=2\cdot 2^{m_{1}}
+1$. $m_{1}$ is odd, then
$$
(2^{\frac{m_{1}+1}{2}})^2
< 2\cdot 2^{m_{1}} +1 < (2^{\frac{m_{1}+1}{2}}
+1)^2.
$$
$2\cdot 2^{m_{1}}+1$ is not a sqare. Then
$f_{a,0}$ is not a hyper-bent function.

When $b=1$, from Proposition \ref{prop3.6},
$\Lambda(a,1)
=2S_{0}-\Lambda(a,0)$. From Proposition
\ref{prop3.7},
$$
\Lambda(a,1)=-\frac{3}{5}(1-K_{m}(a))+
\frac{4}{5}Q_{m}(a)]
=-\frac{1}{5}[3(1-K_{m}(a))-4Q_{m}(a)].
$$
From Proposition \ref{prop3.1},
$f_{a,1}$ is a hyper-bent function if and only if  $3(1-K_{m}(a))-4Q_{m}(a)=-5$.
From Lemma \ref{lem3.3},
\begin{equation}\label{equation4}
3(1-K_{m_{1}}(a))^2=6\cdot2^{m_{1}}+4Q_{m}(a)
-5.
\end{equation}
Note that $Q_{m}(a)\in\{
0, 2\cdot 2^{m_{1}}, -4\cdot 2^{m_{1}}
\}$.
If $Q_{m}(a)=0$, $(\ref{equation4})$ does not hold. If $Q_{m}(a)=2\cdot 2^{m_{1}}$, then
\begin{equation}\label{equation5}
3(1-K_{m_{1}}(a))^2=14\cdot2^{m_{1}}
-5.
\end{equation}
From Proposition \ref{prop2.6},
$1-K_{m_{1}}(a)
\in [-2\cdot 2^{m_{1}}, 2\cdot 2^{m_{1}}]$. Since $m_{1}\geq 3$, $(\ref{equation5})$ does not hold.
If $Q_{m}(a)=-4\cdot 2^{m_{1}}$, then
\begin{equation}\label{equation6}
3(1-K_{m_{1}}(a))^2=-10\cdot2^{m_{1}}
-5,
\end{equation}
Obviously, $(\ref{equation6})$ does not hold.
As a result, $f_{a,1}$ is not a hyper-bent function.

When $b\in \{\beta+\beta^2, \beta+
\beta^3, \beta^2+\beta^4, \beta^3+
\beta^4, \beta+\beta^4, \beta^2+\beta^3
\}$, $\Lambda(a,b)=S_{0}$.
From Proposition \ref{prop3.7},
$
\Lambda
(a,b
)=\frac{1}{5}[1-K_{m}(a)+2Q_{m}(a)]$.
If $f_{a,b}$ is a hyper-bent function, then
$\Lambda(a,b)=1$. Hence,
$1-K_{m}(a)=5-2Q_{m}(a)$.
Note that $Q_{m}(a)\in \{0, 2\cdot2^{m_{1}},
-4\cdot 2^{m_{1}}\}$. Since $1-K_{m}(a)\in [-2\cdot2^{m_{1}}, 2\cdot2^{m_{1}}]$,
$Q_{m}(a)=0$ and  $1-K_{m}(a)=5$.
Then $K_{m}(a)=-4$. From Corollary \ref{cor3.3},
$K_{m}(a)\neq -4$, which gives a contradiction.
Hence, $f_{a,b}$ is not a hyper-bent function.

When $b\in\{1+\beta, 1+\beta^2,
1+\beta^3, 1+\beta^4\}$,
$\Lambda(a,b)=-\frac{3S_{0}-\Lambda(a,0)}{2}$.
From Proposition \ref{prop3.7}, $\Lambda(a,b)=\frac{1}{5}[1-K_{m}(a)
-3Q_{m}(a)]$. If $f_{a,b}$ is a hyper-bent function,
$1-K_{m}(a)=3Q_{m}(a)+5$. Further, we have
$Q_{m}(a)=0$ and $1-K_{m}(a)=5$, that is,
$K_{m}(a)=-4$, which contradicts Corollary \ref{cor3.3}. Hence, $f_{a,b}$ is not a hyper-bent function.

Above all, this theorem follows.
\end{proof}

We now introduce some results on generalized
Ramanujan-Nagell equations \cite{BSh,Le1,Le2}.
A generalized Ramanujan-Nagell equation is of the form
$$
D_{1}x^2+D_{2}=\eta^2\cdot p^k,
$$
where $D_{1}$, $D_{2}$ are two positive integers, $p$ is a prime and $\eta\in
\{1,\sqrt{2},2\}$.

Generally, a generalized Ramanujan-Nagell equation have no more that a solution. Some exceptions are listed in \cite{BSh}. In particular, we have the following lemma.
\begin{lemma}\label{lem3.4}
The equation $15x^2+1=2\cdot 2^k$ has only a solution $(x,k)=(1,3)$. The equation $3x^2+5=4\cdot 2^k$ has three solutions $\{(x,k)|(1,1), (3,3), (13,7)\}$.
\end{lemma}

\begin{theorem}\label{thm3.4}
Let $n=2m$, $m=2m_{1}$,
$m_{1}\equiv 1\pmod 2$ and  $m_{1}\geq 3$.
If $n\neq 12, 28$, then the Boolean functions
$f_{a,b}$ in
$\mathcal{H}_{n}$ of the form
$$
\mathrm{Tr}_{1}^{n}(ax^{2^m-1})+
\mathrm{Tr}_{1}^{4}(bx^{\frac{2^n-1}{5}}),
a\in \mathbb{F}_{2^{m_{1}}},
b\in \mathbb{F}_{16},
$$
is not a hyper-bent function. Further, if $n=12$,
all the hyper-bent functions with $a\in \mathbb{F}_{2^{3}}$ in $\mathcal{H}_{12}$ are  $\mathrm{Tr}_{1}^{12}(ax^{2^6-1})
+\mathrm{Tr}_{1}^{4}(bx^{\frac{2^{12}-1}{5}})$, where $(a+1)(a^3+a^2+1)=0$ and $b=\beta^{i}, i=1,2,3,4$.
If $n=28$, all the hyper-bent functions with $a\in \mathbb{F}_{2^{7}}$ in  $\mathcal{H}_{28}$ are  $\mathrm{Tr}_{1}^{28}(ax^{2^{14}-1})
+\mathrm{Tr}_{1}^{4}(bx^{\frac{2^{28}-1}{5}})$, where $(a+1)(a^7+a^6+a^5+a^4+a^3+a^2+1)=0$ and
$b=\beta^{i}, i=1,2,3,4$.
\end{theorem}
\begin{proof}
If $b=\beta^i~(i=1,2,3,4)$, then
$\Lambda(a,b)=-\frac{S_{0}+
\Lambda(a,0)}{2}$.
From Proposition \ref{prop3.7},
\begin{equation*}
\Lambda(a,b)= -\frac{1}{5}[3(1-K_{m}(a))+Q_{m}(a)].
\end{equation*}
Then $f_{a,b}$ is a hyper-bent function if
and only if $\Lambda(a,b)=1$, that is,
\begin{equation}\label{equation7}
3(1-K_{m}(a))+Q_{m}(a)=-5.
\end{equation}
From Lemma \ref{lem3.3},
$Q_{m}(a)\in\{
0, 2\cdot 2^{m_{1}}, -4\cdot
2^{m_{1}}\}$. If $Q_{m}(a)=0$,
(\ref{equation7}) does not hold.

From Lemma \ref{lem3.3},
\begin{equation}\label{equation8}
3(1-K_{m_{1}}(a))^2=6\cdot 2^{m_{1}}-Q_{m}(a)-5.
\end{equation}
If $Q_{m}(a)=2\cdot 2^{m_{1}}$,
\begin{equation}\label{equation9}
3(1-K_{m}(a))^2+5=4\cdot2^{m_{1}}.
\end{equation}
From Lemma \ref{lem3.4}, we have
$(K_{m_{1}}(a), m_{1})
\in\{(4,3), (-2,3), (14,7),(-12,7)\}$.
Since $4|K_{m_{1}}(a)$, $(K_{m_{1}}(a), m_{1})
=(4,3)$ or $(-12,7)$.

If $Q_{m}(a)=-4\cdot 2^{m_{1}}$, then
$3(1-K_{m_{1}}(a))^2=5(2\cdot 2^{m_{1}}-1)$.
Hence,  $5|1-K_{m_{1}}(a)$. Then
\begin{equation}\label{equation11}
2\cdot2^{m_{1}}=15(
\frac{1-K_{m_{1}}(a)}{5})^2+1.
\end{equation}
From Lemma \ref{lem3.4}, $(K_{m_{1}}(a), m_{1})
=(6,3)$ or $(-4,3)$.
Since $4|K_{m_{1}}(a)$, $(K_{m_{1}}(a), m_{1})
=(-4,3)$.

If $(K_{m_{1}}(a), m_{1})
=(4,3)$, then $Q
_{m}(a)=2\cdot 2^{m_{1}}$. If $(K_{m_{1}}(a), m_{1})=(-4,3)$, then $Q
_{m}(a)=-4\cdot 2^{m_{1}}$.
From Lemma \ref{lem3.3}, if $m_{1}=3$, then
$n=12$ and $f_{a,\beta^i}~(a\in\mathbb{F}_{8},
i=1,2,3,4)$ is a hyper-bent function if and only if one of the assertions $(1)$ and $(2)$ holds.

{\rm (1)} $p(x)=x^5+x+a^{-1}
=(2)(3)$ and  $K_{3}(a)=4$.

{\rm (2)} $p(x)=x^5+x+a^{-1}
=(1)^3(2)$ and $K_{3}(a)=-4$.

If $(K_{m_{1}}(a), m_{1})
=(-12,7)$, then $Q
_{m}(a)=2\cdot 2^{m_{1}}$.
From Lemma \ref{lem3.3}, if
$m_{1}=7$, then $n=28$ and $f_{a,\beta^i}~(a\in\mathbb{F}_{128}, i=1,2,3,4)$ is a hyper-bent function if and only if $p(x)=x^5+x+a^{-1}
=(2)(3)$ and $K_{7}(a)=-12$.

With the help of experiments on the computer,
if $m_{1}=3$, then $n=12$ and all the hyper-bent  functions with $a\in\mathbb{F}_{8}$ in  $\mathcal{H}_{12}$ are
$$
\mathrm{Tr}_{1}^{12}(ax^{2^6-1})+
\mathrm{Tr}_{1}^{4}(bx^{\frac{2^{12}-1}{5}}),
$$
where $(a+1)(a^3+a^2+1)=0$ and $b=\beta^i(i=1,2,3,4)$.

If $m_{1}=7$, then $n=28$ and all the hyper-bent functions with $a\in\mathbb{F}_{128}$ in
$\mathcal{H}_{28}$ are
$$
\mathrm{Tr}_{1}^{28}(ax^{2^{14}-1})+
\mathrm{Tr}_{1}^{4}(bx^{\frac{2^{28}-1}{5}}),
$$
where $(a+1)(a^7+a^6+a^5+a^4+a^3+a^2+1)=0$ and $b=\beta^i~(i=1,2,3,4)$.

Above all, this theorem follows.
\end{proof}

\section{Conclusion}
This paper considers
the hyper-bentness of the Boolean functions
$f_{a,b}$ of the form
$f_{a,b}:=
\mathrm{Tr}_{1}^{n}(ax^{2^{m}-1})+
\mathrm{Tr}_{1}^{4}(bx^{\frac{2^{n}-1}{5}})$,
where $n=2m$, $m=2\pmod 4$, $a\in \mathbb{F}_{2^m}$ and $b\in\mathbb{F}_{16}$.
If $b=1$ or $b$ is a primitive element
in $\mathbb{F}_{16}$ such that $\mathrm{Tr}_{1}^{4}(b)=0$, the hyper-bentness
of $f_{a,b}$ can be characterized by
Kloosterman sums and the factorization of $x^5+x+a^{-1}$. If $a\in\mathbb{F}_{2^{\frac{m}{2}}}^*$, with the help of generalized Ramanujan-Nagell equations, we
prove that $f_{a,b}$ is not a hyper-bent function unless $n=12$ or $n=28$. Further, we give all the
hyper-bent functions for $n=12$ or $n=28$.



\begin{thebibliography}{}
%
\bibitem{BSh}  Bugeaud Y.,  Shorey T.N.;  On the number of solutions of the
generalized Ramanujan-Nagell equation, I.J. reine angew. Math. 539
(2001), 55-74.


\bibitem{APG}  Canteaut A.,  Charpin P., Kyureghyan
G.: A new class of monomial bent functions, Finite Fields Applicat., vol. 14, no. 1, pp
221-241, 2008.

\bibitem{CNa}   Cardona G.,   Nart E.: Zeta function and cryptographic exponent of supersingular
curves of genus 2. In Pairing-Based Cryptography | Pairing 2007, volume 4575 of
Springer Lecture Notes in Computer Science, pages 132-151, 2007. 46


\bibitem{CC} Carlet C.: Boolean functions for cryptography and error correcting
codes, in Chapter of the Monography ¡°Boolean Models and Method
in Mathematics, Computer Science, and Engineering¡±, Y. Crama and
P. L. Hammer, Eds. Cambridge, U.K.: Cambridge Univ. Press, 2010
pp. 257-397.

\bibitem{CP} Carlet  C.,  Gaborit P.: Hyperbent functions and cyclic codes, J
Combin. Theory, ser. A, vol. 113, no. 3, pp. 466-482, 2006.

\bibitem{CTA}   Carlet C.,  Helleseth T., Kholosha A.,   Mesnager S.: On the Dual of the Niho Bent Functions with $2^r$  Exponents, 2011, to be published.

\bibitem{CS}  Carlet C.,   Mesnager S.: On Dillon's Class H of Bent Functions
Niho Bent Functions and O-Polynomials Cryptology ePrint Archive
Report no 567.



\bibitem{PG}  Charpin P., Gong G.: Hyperbent functions, Kloosterman sums and
Dickson polynomials, IEEE Trans. Inf. Theory, vol. 9, no. 54, pp
4230-4238, 2008.





\bibitem{PGK}   Charpin P.,   Kyureghyan G.: Cubic monomial bent functions: A
subclass of $\mathcal{M}$ , SIAM J. Discr. Math., vol. 22, no. 2, pp. 650-665
2008.

\bibitem{PEC}  Charpin P., Pasalic E., Tavernier C., On bent and semi-ben
quadratic Boolean functions, IEEE Trans. Inf. Theory, vol. 51, no
12, pp. 4286-4298, 2005.



\bibitem{JD}  Dillon J.:, Elementary Hadamard Difference Sets, Ph.D., Univ.Mary-
land,  1974.

\bibitem{JH} Dillon  J. F.,  Dobbertin H.: New cyclic difference sets with Singer
parameters, Finite Fields Applicat., vol. 10, no. 3, pp. 342-389, 2004.



\bibitem{HGT}  Dobbertin H.,  Leander G.: A survey of some recent results on bent functions, in T. Helleseth et al. (eds.) Sequences and Their Applications, LNCS 3486, pp. 1-29, Springer,
Heidelberg, 2004.

\bibitem{HGA}   Dobbertin H.,  Leander G.,  Canteaut A.,  Carlet C.,  Felke P.,   Gaborit P.: Construction of bent functions via Niho power functions, J.
Combin. Theory, ser. A, vol. 113, pp. 779-798, 2006.

\bibitem{RG}  Gold R.: Maximal recursive sequences with 3-valued recursive
crosscorrelation functions, IEEE Trans. Inf. Theory, vol. 14, no. 1,
pp. 154-156, 1968.



\bibitem{GS} Gong G.,  Golomb S. W.: Transform domain analysis of DES,
IEEE Trans. Inf. Theory, vol. 45, no. 6, pp. 2065-2073, 1999.

\bibitem{Hel}  Helleseth T.,  Zinoviev V.A.:  On $Z_4$-Linear Goethals Codes and
Kloosterman Sums, Designs, Codes and Cryptography, vol. 17, no.
1-3, pp. 246-262, 1999.


\bibitem{HD}  Hu H., Feng D.: On quadratic bent functions in polynomial forms,
IEEE Trans. Inf. Theory, vol. 53, no. 7, pp. 2610-2615, 2007.

\bibitem{TK}  Kasami T.: Weight enumerators for several classes of subcodes of the
2nd-order Reed¨CMuller codes, Inf. Contr., vol. 18, pp. 369-394, 1971.

\bibitem{SJ}   Kim S. H.,  No J. S.: New families of binary sequences with low
correlation, IEEE Trans. Inf. Theory, vol. 49, no. 11, pp. 3059-3065,
2003.

\bibitem{LWo}  Lachaud  G.,   Wolfmann J.: The weights of the orthogonal of the extended quadratic binary Goppa codes, IEEE Trans. Inform. Theory, 36 (1990), pp. 686-692

\bibitem{Le1}  Maohua Le: On the Diophantine Equations $d_1x^2+ d_2 = 4y^n$, Proc. Amer. Math.Soc, 118 (1993), 67-70.

\bibitem{Le2}  Maohua Le,  Xu T.: On the Diophantine Equation $D_1x^2 +D_2 = k^n$, Publ. Math. Debrecen, 47 (1995), 293-297.





\bibitem{GL}   Leander G.: Monomial bent functions, IEEE Trans. Inf. Theory, vol.
2, no. 52, pp. 738-743, 2006.

\bibitem{GA} Leander G., Kholosha A.: Bent functionswith   Niho exponents,
IEEE Trans. Inf. Theory, vol. 52, no. 12, pp. 5529-5532, 2006.

\bibitem{WMF}   Ma W., Lee M.,  Zhang F.: A new class of bent functions, IEICE
Trans. Fund., vol. E88-A, no. 7, pp. 2039-2040, 2005.


\bibitem{RLM}   McFarland R. L.: A family of noncyclic difference sets, J. Combin.
Theory, ser. A, no. 15, pp. 1-10, 1973.

\bibitem{MNa}  Maisner D., Nart E.:  Zeta functions of supersingular curves of genus 2.
Canadian Journal of Mathematics 59, 372-392 (2007)

\bibitem{SM3}  Mesnager S.: A new class of bent boolean functions in polynomial
forms, in Proc. Int. Workshop on Coding and Cryptography, WCC
2009, 2009, pp. 5-18.

\bibitem{SM1} Mesnager S.: A new class of bent and hyper-bent boolean functions
in polynomial forms, Des. Codes Cryptography, 59(1-3):265-279, 2011

\bibitem{SM}   Mesnager S.: Bent and Hyper-Bent Functions in Polynomial Form and Their Link With Some Exponential Sums and Dickson Polynomials, IEEE Trans. Inf. Theory, vol. 57, no. 9, pp. 5996-6009, 2011

\bibitem{SMT}  Mesnager S.: Hyper-bent boolean functions with multiple trace terms. In M. Anwar Hasan and Tor Helleseth, editors, WAIFI, volume 6087 of Lecture Notes in Computer Science, pages 97-113. Springer, 2010.


\bibitem{SMG} Mesnager S.: A new family of hyper-bent Boolean functions in polynomial form. Proceedings
of Twelfth International Conference on Cryptography and Coding, Cirencester, United Kingdom.
M. G. Parker (Ed.): IMACC 2009, LNCS 5921, pp 402-417, Springer, Heidelberg (2009).


\bibitem{GRG}  Mullen G. L.,  Lidl R.,  Turnwald G.:  Dickson Polynomials. Reading, MA: Addison-Wesley, 1993, vol. 65, Pitman
Monographs in Pure and Applied Mathematics.

\bibitem{OSR}   Rothaus O. S.: On bent functions, J. Combin. Theory, ser. A, vol. 20,
pp. 300-305, 1976.




\bibitem{vv}  van der Geer G.,   van der Vlugt M.:  Reed-Muller  codes  and  supersingular curves,  I,  Compositio  Math.  84  (1992),  333-367.

\bibitem{AG}  Youssef A. M., Gong G.:  Hyper-bent functions, in Advances in
Crypology¨CEurocrypt¡¯01, 2001, LNCS, pp. 406-419.

\bibitem{NG}   Yu N. Y.,  Gong G.: Construction of quadratic bent functions
in polynomial forms, IEEE Trans. Inf. Theory, vol. 7, no. 52, pp.
3291-3299, 2006.

\end{thebibliography}


\end{document}